\documentclass[12pt]{amsart}
\usepackage{latexsym,amsmath,amsfonts,amscd,amssymb}
\usepackage{graphicx}
\usepackage{enumerate}
\newtheorem{theorem}{Theorem}[section]
\newtheorem{theorem*}{Theorem}
\newtheorem{lemma}[theorem]{Lemma}
\newtheorem{proposition}[theorem]{Proposition}
\newtheorem{proposition*}[theorem*]{Proposition}
\newtheorem{corollary}[theorem]{Corollary}
\newtheorem{corollary*}[theorem*]{Corollary}

\newtheorem{definition}[theorem]{Definition}

\theoremstyle{remark}

\newtheorem{remark*}[theorem*]{Remark}

\newtheorem{note*}[theorem*]{Note}

\newcommand{\EE}{{\mathbb E}}

\newcommand{\PP}{{\mathbb P}}

\newcommand{\RR}{{\mathbb R}}

\title{On profitability of selfish mining}


\author[C. Grunspan]{Cyril Grunspan}
\address{Cyril Grunspan\newline{}\indent L\'eonard de Vinci P\^ole Univ, Finance Lab, Labex R\'efi\newline{}\indent Paris, France, }
\email{cyril.grunspan@devinci.fr}
\author[R. P\'{e}rez-Marco]{Ricardo P\'{e}rez-Marco}
\address{Ricardo P\'{e}rez-Marco\newline{}\indent CNRS, IMJ-PRG, Labex R\'efi \newline{}\indent Paris, France}
\email{ricardo.perez.marco@gmail.com}
\address{\tiny Author's Bitcoin Beer Address (ABBA)\footnote{\tiny Send some bitcoins to support our research at the pub.}:\newline{}\indent 1KrqVxqQFyUY9WuWcR5EHGVvhCS841LPLn} 

\address{\includegraphics[scale=0.5]{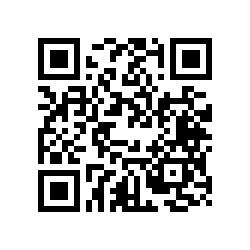}}

\begin{document}

\begin{abstract}
  We review the so called selfish mining strategy in the Bitcoin network and compare its profitability to honest mining.
  We build a rigorous profitability model for repetition games. 
  The time analysis of the attack has been ignored in the previous literature based on a Markov model,
  but is critical. Using martingale's techniques and Doob Stopping Time Theorem we compute the expected duration of attack cycles. 
  We discover a remarkable property of the bitcoin network:  no strategy is more profitable 
  than the honest strategy before a difficulty adjustment. So selfish mining can only become profitable afterwards, thus it is 
  an attack on the difficulty adjustment algorithm. We propose an improvement of  
  Bitcoin protocol making it immune to selfish mining attacks. We also study miner's attraction to selfish mining pools. 
  We calculate the expected duration time before profit for the selfish miner, a computation that is out of reach by the 
  previous Markov models.
\end{abstract}

{\maketitle}

\section{Introduction}

The stability of Bitcoin protocol \cite{N08} relies on rules aligned with self-interest of participants in the network. One 
rule is that miners make public blocks as soon as they are validated. 
``Selfish Mining'' is a deviant mining strategy described in {\cite{ES14}} where a miner 
withholds validated blocks and releases them with a well timed strategy designed to invalidate the maximum number of blocks 
mined by the rest of the network. 


Other researchers have proposed other 
selfish mining strategies which are supposed to be ``optimal'' {\cite{SSZ16}}. The selfish mining strategy is 
presented in courses and textbooks on Bitcoin such  as  {\cite{BFGMN16}} 
or {\cite{W17}}.


All these articles do not make a proper analysis of the profitability of the attack compared to honest mining, and, more critically, do 
ignore time considerations. More precisely, 
the Markov model used in these papers is limited by inception since it does not incorporate an analysis of 
the time duration of the attack. The main goal of this article 
is to carry out a proper analysis of profitability that is lacking in the literature. It turns out that without difficulty adjustments 
the strategy is unsound.

\medskip

\section{Selfish mining strategy.} \label{sec_SM}

We describe the selfish mining strategy presented in  {\cite{ES14}}. The selfish miner attack starts by validating 
and not broadcasting a block, then continuing mining secretly on top of this block. Then he proceeds as follows:

\begin{enumerate}
  \item \label{case1} If the advance of the selfish miner is $1$ block and the
  honest miners discover a block then the selfish miner broadcasts immediately
  the block he has mined secretly. A competition then follows. The selfish miner mines on top of his now public block. The selfish miner 
  is sufficiently well connected with the rest of the network
  so that a fraction $0\leq\gamma \leq 1$ of the honest network accepts his block
  proposal and starts mining with him on top of it.
  

  \item \label{compet}If the advance of the selfish miner is $2$ blocks and
  the honest miners discover a block, then the selfish miner broadcasts
  immediately all the blocks he has mined secretly. Then, the whole network
  switches to his fork.
  
  
  \item \label{case3} If the advance of the selfish miner is greater than $2$, as soon as the honest miners discovers one block, 
  then the selfish miner makes public one more block releasing a subchain that ends with that block that enters into competition 
  with the new honest block \footnote{It is not enough to release only block. Line 26 of Algorithm 1 in \cite{ES14} is not accurate.}. 
  The selfish miner keeps mining on top of his secret chain.
  
  
  \item Except in (1), the selfish miner keeps on mining secretly on top of his
  fork.
\end{enumerate}

Note that if the advance of the selfish miner is greater than $2$, then at some
point his advance will be equal to $2$ (because we assume his
hashrate to be less than $50\%$, or other more efficient attacks are possible) and 
then, according to the second point, the whole
network ends up adopting the fork proposed by the selfish miner.
Therefore, the blocks made public by the selfish miner when his advance is
greater than $2$ always end up being accepted by the network. Point (3) is somewhat 
irrelevant since the only thing that counts when the selfish miner takes an advantage is to force his validated blocks in the public 
blockchain. He can ignore block validations by honest miners, except when the advance is only $2$ and then release the whole secret fork.

In  {\cite{ES14}} it is assumed that the fraction $\gamma$ stays always constant.
This is not accurate since $\gamma$ depends on the timing of the discovery of a new block mined 
by the network, and therefore cannot be constant.
But it is necessary to assume it constant for the sake of the Markov chain model presented in {\cite{ES14}} 
and these authors made such assumption.  
The analysis of the necessary capital to 
reach a stable regime is not done. But, more importantly, the difference in profitability by deviating from 
the Bitcoin protocol are not properly accounted. This is fundamental in order to compute the profitability of such a rogue 
strategy. The time analysis is also critical to estimate the profitability and is ignored.

\medskip

\section{Profitability of selfish vs. honest mining.}

\subsection{Cost of mining.}
The key idea in order to evaluate properly the profitability of selfish mining is to compare it to the profitability 
of honest mining. We assume that miners are in the mining business because the operating cost (equipment, energy power, salaries, hedge on the volatility of
bitcoin exchange rate, etc)
are compensated by the block reward of newly generated bitcoins plus transaction 
fees\footnote{There may be other non-economic incentives, as transforming non-internationally 
circulating currency into bitcoins, that we cannot consider.}. 

In the article, we assume that the selfish miner's hashrate is constant. His strategy determines if he should  
release a block as soon as he has just validated one or if he should keep mining secretly on a private fork and 
release it at the appropriate timing. Whatever the strategy is, the machines of the miner (in this case ASICs)  
operate at full capacity repeating non-stop the same hashing calculations. Thus the 
cost of mining per unit time of any block withholding strategy is the same as the one for the honest strategy .
 
If $c_S(t)$ (resp. $c_H(t)$) is the random variable measuring the instantaneous cost of mining per unit time for the selfish (resp. honest) mining strategy
at instant $t$. We have just proved that $c_S(t)=c_H(t)=c(t)$. For the short 
duration of the attack we can assume that market conditions are stable, so that the random variables $(c(t))_t$ are i.i.d. 
We also assume that the random variables $(c(t))$ are integrable on $t\in \RR$ since $\int_{t_1}^{t_2} c(t) \, dt$ is the random variable giving the cost 
during the period $[t_1,t_2]$ that is finite.
Note that the cost $c(t)$ is not only independent of the strategy but also
of the block mined being accepted or not by the network. Miners get a profit in bitcoin and have expenses in 
local currencies, to pay for equipment, salaries, etc
which makes mining a risky business. Costs for hedging against this risk is independent of the strategy and is integrated in $c(t)$. 
The comparative analysis we carry out is independent of the bitcoin exchange rate fluctuations, that are independent of the mining 
strategy (therefore we can evaluate $c(t)$ in bitcoin per second for the instantaneous exchange rate). 

\medskip

\subsection{Profit and Loss.} Our goal is to evaluate the profit and loss (PnL) and compare the PnL of selfish mining 
and honest mining. Profit and loss
is the revenue $R$ minus the costs $C$,
$$
PnL=R-C \ .
$$
The revenue in the mining business comes from the block reward $b>0$ that includes the bounty $b_0$ on newly created 
bitcoins\footnote{At this time, $b_0=12.5 \text{ BTC}$.} plus the 
total amount of fees $f$ of the transactions included in the block, once this block has been accepted by the official blockchain.
The amount of fees $f$ is different from one block to another. 
We assume that we are not near a halving, so that the reward $b_0$ stays constant. 
As we prove below, we can assume for the profitability analysis that $b>0$ is constant and equal to its expected value $b=b_0+\EE[f]$.
Double spends may also increase the profitability of a dishonest strategy.

It is important to note that, for any business, is the $PnL$ per unit of time, and not the $PnL$ per block solved or accepted by the network
that counts. It is also important to observe that PnL per block or per unit of time are not equivalent since the 
strategy employed does delay the speed of validation of blocks in 
the network.

\subsection{Profit and Loss per unit of time.} A non-stop attack strategy, as selfish mining, consists in a consecutive sequence of ``attack cycles''. 

In the 
case of selfish mining it starts when both selfish miners and honest miners are working on top of the same public blockchain. The selfish miners
aim to get an advance of their secretly mined blockchain. If, at the beginning, the honest miners find the new block, then a new cycle starts. If
the selfish miners succeed in building an advantage, then the cycles lasts until the honest miners catch-up and force the selfish miners to release
all of their secretly mined blocks. Then a new cycle starts. 

\medskip

For this type of strategies of games with repetition,  the asymptotic $PnL$ per unit of time, $PnLt_\infty$, can be evaluated. 
This is the content of the following Theorem. 

\bigskip

\begin{theorem}[Profitability of an attack]\label{thm_profit}
 Let $R$, $C$ and $T$ be the random variables corresponding respectively to the revenue, cost and duration of a cycle for a repetition strategy. We assume 
 these random variables to be integrable. Then in the long run the $PnL$ per unit 
 of time of the repeated strategy is
 $$
 PnLt_\infty =\frac{\EE[R]-\EE[C]}{\EE[T]} \ .
 $$
 We call $PnLt_\infty$ the profitability of the strategy.
\end{theorem}

\begin{proof}
Let $R_i$, $C_i$ and $T_i$, be the corresponding values for the $i$-cycle. The $(R_i)$ (resp. $(C_i)$, $(T_i)$) are i.i.d. random variables.
The $PnL_n$ after $n$-cycles is given by
$$
PnL_n = \frac{\sum_{i=1}^n R_i -\sum_{i=1}^n C_i}{\sum_{i=1}^n T_i} = \frac{\frac{1}{n}\sum_{i=1}^n R_i -\frac{1}{n}\sum_{i=1}^n C_i}
{\frac{1}{n}\sum_{i=1}^n T_i}\ .
$$
By the strong law of large numbers \cite{R12} we have that almost surely
$$
\lim_{n\to +\infty } PnL_n = \frac{\EE[R]-\EE[C]}{\EE[T]} \ .
$$
\end{proof}

We consider integrable games with random variables $R$, $C$, and $T$ are all integrable.

\begin{definition}[Integrable games]
A game or strategy is integrable when $R$, $C$, and $T$ are integrable.
\end{definition}

We prove later that the selfish mining strategy is integrable. 
%

\begin{proposition}
For an integrable game with cost per unit time $c$, we have
$$
\frac{\EE[C]}{\EE[T]} = \EE[c] 
$$
\end{proposition}

\begin{proof}
We have 
$$
C=\int_0^T c(t) \, dt 
$$
and, since $T$ is integrable and the $(c(t))$ are integrable, i.i.d. and independent with $T$, by Wald's identity we have
$$
\EE[C] = \EE[T] .\EE[c]  \ .
$$
\end{proof}

\begin{definition}[Revenue ratio]
The \textit{revenue ratio} of a strategy $\xi$ is defined as
$$
\Gamma(\xi) =\frac{\EE[R]}{\EE[T]} \ .
$$
\end{definition}

\begin{definition}[Cost ratio]
The \textit{cost ratio} of a strategy $\xi$ is defined as
$$
\Upsilon (\xi) =\frac{\EE[C]}{\EE[T]} \ .
$$
\end{definition}

For selfish mining and honest mining we have equal cost ratio $\EE[C]/\EE[T]=\EE[c]$.

The \textit{revenue ratio} is the benchmark for profitability of integrable strategies with equal cost ratio, more precisely, we have

\begin{proposition}[Comparison of profitabilities]
We consider two integrable strategies with equal cost ratio. 
The strategy $\xi$ is more profitable than strategy $\xi'$ if and only if
$$
\Gamma(\xi')\leq \Gamma(\xi) \ .
$$
\end{proposition}

\begin{proof}
Since 
$$
\frac{\EE[C(\xi)]}{\EE[T(\xi)]} = \frac{\EE[C(\xi')]}{\EE[T(\xi')]}
$$
we get that 
$$
PnLt_\infty(\xi) - PnLt_\infty(\xi')=\frac{\EE[R(\xi)]-\EE[C(\xi)]}{\EE[T(\xi)]} - \frac{\EE[R(\xi')]-\EE[C(\xi')]}{\EE[T(\xi')]} = \Gamma(\xi)- \Gamma(\xi')
$$
and the result follows.
\end{proof}

\medskip

\medskip

In the previous literature (in particular in {\cite{ES14}} and \cite{SSZ16}  ) the authors 
were only considering without proper justification 
a ``relative revenue'' benchmark for profitability. This is a non-standard 
notion in accountability that they define as the ratio
\begin{eqnarray*}
 \frac{\mathbb{E} [R_S]}{\mathbb{E} [R_S]+\mathbb{E} [R_N]}
\end{eqnarray*}
where $R_S$ and $R_N$ are the revenues of the selfish miners and the rest of the network. 
The effect of the selfish mining strategy 
is to reduce both $\mathbb{E} [R_S]$ 
and $\mathbb{E} [R_N]$, and to increase this relative revenue on certain conditions on $q$ and $\gamma$. 
Obviously, increasing the relative revenue 
at the cost of reducing its own revenue is not a sound strategy in general. Also discussing profitability via a Markov model without 
time duration considerations of the strategy is unsound. The time dynamics is absent in the Markov model, and cannot take 
account situations where the expected time of some attack cycles may take very long time and will impact the $PnLt$. In particular, 
as we prove, this is what makes selfish mining non-profitable without a difficulty adjustment. But after a difficulty adjustment,
and for attack cycles small compared to the difficulty adjustment period, the ``relative revenue'' converges to our Revenue Ratio and 
so we can justify their benchmark. Also in \cite{SSZ16} the ``objective function'' $REV$ that these authors use, that is an 
asymptotic version of the ``relative revenue'' above, in their definition at the end of section 2 we can remove the expectation as well 
as the 
$\liminf$ in the definition formula since the limit exists and is constant by the same argument that we have used in the proof of 
Theorem \ref{thm_profit} (by using the strong law of numbers Theorem provided one can prove that the strategies they consider are integrable).

\subsection{Block rewards and transaction fees.}

As we already explained, the revenue comes from the block reward $b=b_0+f$. We have that $b_0$ is constant and $f$ is a random variable.
Let $Z$ the random variable denoting the number of blocks validated by the miner in an attack cycle.

\begin{proposition}
We have
$$
\EE[R]= \EE[Z] (b_0 + \EE[f])=\EE[Z] \EE[b]
$$
\end{proposition}

\begin{proof}
The revenue after an attack cycle consists in the sum of the rewards $(b_i)_{1\leq i\leq Z}$ of the 
blocks validated by the miner in the official blockchain,
$$
R=\sum_{i=1}^Z b_i =\sum_{i=1}^Z (b_0 +f_i)=Z b_0 +\sum_{i=1}^Z f_i
$$
The random variables $(f_i)$ are i.i.d. and also independent with $Z$. So we have by Wald's identity
$$
\EE\left [\sum_{i=1}^Z f_i \right ]=\EE[Z].\EE[f]
$$
so,
$$
\EE[R] = \EE[Z] (b_0 + \EE[f]).
$$
\end{proof}

Therefore, for the purpose of computation of $\EE[R]$, nothing changes if we consider $b$ constant equal to 
its average value (that is the bounty plus the average of total transaction fees per block). So we can assume without loss of generality 
that the reward per block $b$ is constant.

\section{Mining strategies.}

Let us fix some notations. We have two sets of miners (for example, honest miners and attacking miners for an attack). 
The progression of blocks are described by two independent Poisson processes $N$ and $N'$ respectively, 
with parameter $\alpha$, resp. $\alpha'$.
Interblock validation times are denoted by $(T_i)$ for the honest miners and $(T'_i)$ for the selfish miner. We denote 
for $n\geq 1$
\begin{align*}
 S_n &= T_1+T_2+\ldots +T_n  \ ,\\
 S'_n &= T'_1+T'_2+\ldots +T'_n \ .\\
\end{align*}
We recall (see \cite{GPM17}) that the random variables $(T_i)$ (resp. $(T'_i)$) follow an exponential density with parameter $\alpha$ (resp. $\alpha'$) 
and the random variable $S_n$ (resp. $S'_n$) follows a gamma density distribution with parameter $(n,\alpha)$ (resp. $(n,\alpha')$). We denote 
\begin{align*}
 \tau_0 &= \frac{1}{\alpha + \alpha'} \ ,\\
 p &= \frac{\alpha}{\alpha + \alpha'} \ ,\\
 q &= \frac{\alpha'}{\alpha + \alpha'} \ ,
\end{align*}
We have $p+q=1$ and if $\alpha'<\alpha$ then $0<q<1/2<p<1$ (and $\tau_0 =10 \text{ min}$ for 
the Bitcoin network in normal conditions). The quantities $p$ and $q$ 
represent relative hashrates, and also probabilities of finding the next block by each group 
of miners as the following elementary computations show \cite{R12}:

\begin{lemma}
 We have
 \begin{align*}
  \PP [T_1 < T'_1] &= p \ , \\
  \PP [T'_1 < T_1] &= q \ , \\
  \PP [T'_1 <T_1<S'_2] &= pq \ ,\\
  \EE[T_1\wedge T'_1] &= \tau_0 \ .
 \end{align*}
\end{lemma}

\subsection{Honest strategy.}

The cycle for the honest strategy lasts until a block is found. So the stopping time for the honest strategy is
$$
\tau_H = T_1'\wedge T_1 \ .
$$
We have $\EE[\tau_H] =\tau_0$. 

The reward in a cycle is $0$ or $b$ depending who the miner is, thus 
$$
\EE[R(\tau_H)] =p\cdot 0+q\cdot b =qb\ ,
$$
and we have,

\begin{theorem} \label{honest}
We have that $\tau_H$ and $ R (\tau_H)$ are integrable and 
\begin{align*}
\EE [R (\tau_H)] &= qb  \ ,\\
\EE [\tau_H]  &= \tau_0 \ .
\end{align*}
Therefore,
$$
\Gamma (H)=\frac{\EE [R (\tau_H)]}{\EE [\tau_H]} = \frac{qb}{\tau_0} \ .
$$
\end{theorem}

\subsection{A stability result}

We establish an upper bound for the revenue ratio of a mining strategy using a martingale argument that 
is central for the other results in this article (see the next two sections).

\begin{proposition} \label{upperboundhonest}
Let $\tau$ be the stopping time of an arbitrary integrable strategy. 
Without difficulty adjustments, we have  $\Gamma (\tau)\leq \alpha' b$. 
\end{proposition}

\begin{proof}
For any strategy  we always have $R(\tau)\leq N'(\tau) b$.
By applying Doob's theorem to the compensated martingale $N'(t) - \alpha' t$ 
and to the finite stopping time $\tau\wedge t_0$ with $t_0>0$, we get $\EE[N(\tau\wedge t_0)]\leq \alpha' \EE[\tau\wedge t_0]$. 
Now, using the monotone convergence theorem with $t_0\to +\infty$, we have $\EE[N'(\tau)]\leq \alpha' \EE[\tau]$. 
Hence, $\Gamma (\tau) = \frac{\EE[R(\tau)]}{\EE[\tau]} \leq \alpha' b$.
\end{proof}

Therefore Theorem \ref{honest} implies  that \textit{without a difficulty adjustment}, 
$\Gamma(\tau)\leq \alpha' b = \Gamma (H)$ since $\alpha' = \frac{q}{\tau_0}$, 
so, in this case, honest mining is optimal. 

\begin{theorem}[Theorem from beyond] \label{from_beyond}
Without difficulty adjustments, the honest mining strategy is optimal. 
\end{theorem}

Notice that the difficulty adjustment is necessary during the period of adoption. 
In a steady regime, the hashrate is expected to be roughly constant 
and the difficulty adjustment would be unnecessary and could be removed from the protocol. 
In that situation the honest mining strategy is optimal. This is a remarkable 
and unexpected property of the protocol, that is unlikely to have been foreseen by the creators. 
Just to make the reader think about it, 
we observe this gives some support to the Bitcoin protocol 
being imported from beyond where it is running in a steady state, whereas the name of ``Theorem from beyond''.

\subsection{Selfish mining strategy}\label{sec_selfish}

We consider now the selfish mining strategy as described in Section \ref{sec_SM}. We assume that the hashrate of the attackers is less than that 
of the honest miners (i.e. $\alpha'<\alpha$).
We denote by $\tau_{SM, \gamma}$ the duration time of an attack cycle.

\begin{lemma}
We have
\begin{align*}
  \tau_{SM, \gamma} & = \inf  \{ t\geq T_1 ; N(t) = N'(t) -1 + 2  \cdot \boldsymbol{1}_{T_1 < T'_1} +   2  \cdot \boldsymbol{1}_{T'_1 < T_1<S_2<S'_2} \} \ ,
\end{align*}
and the stopping time $\tau_{SM, \gamma}$ is finite  almost surely.
\end{lemma}

\begin{proof}
Note that if $T_1 < T'_1$, then $ \tau_{SM, \gamma} = T_1$. If $T'_1 < T_1<S_2<S'_2$ then $ \tau_{SM, \gamma} = S_2$. 
If  $T'_1 < T_1<S'_2<S_2$ then $ \tau_{SM, \gamma} = S'_2$. Otherwise we have $S'_2 \leq T_1$ and $N(T_1) = 1 \leq N'(T_1) -1$, and in that case 
 $\tau_{SM, \gamma} = \inf  \{ t\geq T_1 ; N(t) = N'(t) -1 \}$ exists and is finite almost surely since $\alpha' < \alpha$.  
\end{proof}

\begin{theorem}\label{profitsmg}
 We have that $\tau_{SM, \gamma}$ and $R (\tau_{SM, \gamma})$ are integrable, and
 \begin{align*}
\EE [R (\tau_{SM, \gamma})] &=  \frac{(1+pq)(p - q) + pq}{p - q}\, qb -(1-\gamma) p^2q  \, b \ ,\\
\EE [\tau_{SM,\gamma}]  &= \frac{(1+pq) (p - q) + pq}{p - q} \, \tau_0  \ .
\end{align*}
Therefore,
$$
\Gamma(SM,\gamma)=\frac{\EE [R (\tau_{SM, \gamma})]}{\EE [\tau_{SM, \gamma}]} =
  \frac{qb}{\tau_0} -  (1 - \gamma) \frac{p^2 q (p - q)  b}{\left((1+pq)(p - q) + pq \right) \tau_0}\ .
$$
\end{theorem}

\begin{corollary}
  For $\gamma <1$, we have that 
  $$
  \Gamma(SM,\gamma)=\frac{\EE [R (\tau_{SM, \gamma})]}{\EE [\tau_{SM, \gamma}]} < \frac{qb}{\tau_0} = \Gamma(H) \ ,
  $$
  so, the Selfish Mining strategy with $\gamma <1$ is strictly less profitable than the honest strategy. 
\end{corollary}

\begin{proof}[Proof of the Theorem]
For any $t_0\in \RR$, the stopping time $\tau_{SM, \gamma} \wedge t_0$ is bounded. Moreover, the compensated Poisson process 
$N(t) - \alpha t$ (resp. $N'(t) -\alpha' t$) is a well known martingale. So, using Doob's theorem \cite{R12}, we have:
\begin{align*}
&\alpha' \EE [\tau_{SM, \gamma} \wedge t_0]  =  \EE [N' (\tau_{SM, \gamma} \wedge t_0)]\\
  &=   \EE [N' (\tau_{SM, \gamma} \wedge t_0) | \tau_{SM, \gamma} < t_0] \cdot \PP [\tau_{SM, \gamma} < t_0] + 
  \EE [N' (\tau_{SM, \gamma} \wedge t_0) | \tau_{SM, \gamma} > t_0] \cdot \PP [\tau_{SM, \gamma} > t_0] \\
   &=  \EE [N' (\tau_{SM, \gamma}) | \tau_{SM, \gamma} < t_0] \cdot \PP [\tau_{SM, \gamma} < t_0]
  +\EE [N' (t_0) | \tau_{SM, \gamma} > t_0] \cdot \PP [\tau_{SM, \gamma} > t_0]\\
   &=   \EE [N (\tau_{SM, \gamma}) + 1 - 2 (\boldsymbol{1}_{T_1 < T'_1}
  +\boldsymbol{1}_{S'_1 < S_1 < S_2 < S'_2}) | \tau_{SM, \gamma} < t_0] \cdot \PP [\tau_{SM, \gamma} < t_0] \\
  & \ \ \ + \EE [N' (t_0)] \cdot \PP [\tau_{SM, \gamma} > t_0] \\
  &=  \EE [N (\tau_{SM, \gamma}) | \tau_{SM, \gamma} < t_0] \cdot \PP [\tau_{SM, \gamma} < t_0]
  +\PP [\tau_{SM, \gamma} < t_0] \\
  & \ \ \  +  \EE [N' (t_0)] \cdot \PP [\tau_{SM, \gamma} > t_0]  - 2 \cdot \EE [\boldsymbol{1}_{T_1 < T'_1}
  +\boldsymbol{1}_{S'_1 < S_1 < S_2 < S'_2} | \tau_{SM, \gamma} < t_0] \cdot \PP [\tau_{SM, \gamma} <
  t_0]\\
  & =  \EE [N (\tau_{SM, \gamma} \wedge t_0)] -\EE [N (\tau_{SM, \gamma} \wedge t_0) |
  \tau_{SM, \gamma} > t_0] \cdot \PP [\tau_{SM, \gamma} > t_0] +\PP [\tau_{SM, \gamma} < t_0] \\
  & \ \ \ -  2\cdot \EE [\boldsymbol{1}_{T_1 < T'_1} +\boldsymbol{1}_{S'_1 < S_1 < S_2
  < S'_2} | \tau_{SM, \gamma} < t_0] \cdot \PP [\tau_{SM, \gamma} < t_0] + \alpha' t_0 \cdot
  \PP [\tau_{SM, \gamma} > t_0]\\
  & =  \alpha \EE [\tau_{SM, \gamma} \wedge t_0] -\EE [N (t_0)] \cdot
  \PP [\tau_{SM, \gamma} > t_0] +\PP [\tau_{SM, \gamma} < t_0] \\ 
  & \ \ \  -  2 \cdot \EE [(\boldsymbol{1}_{T_1 < T'_1} +\boldsymbol{1}_{S'_1 < S_1 <
  S_2 < S'_2}) \cdot \boldsymbol{1}_{\tau_{SM, \gamma} < t_0}] + \alpha' t_0 \cdot \PP
  [\tau_{SM, \gamma} > t_0]
\end{align*}
and,
\begin{align*}
  (\alpha' - \alpha) \EE [\tau_{SM, \gamma} \boldsymbol{1}_{\tau_{SM, \gamma} > t_0}] & =  
  (\alpha' - \alpha) \EE [\tau_{SM, \gamma} \wedge t_0] - (\alpha' - \alpha) t_0
  \cdot \PP [\tau_{SM, \gamma} > t_0]\\
  &=\PP [\tau_{SM, \gamma} < t_0] - 2\cdot\EE [(\boldsymbol{1}_{T_1 < T'_1}
  +\boldsymbol{1}_{S'_1 < S_1 < S_2 < S'_2}) \cdot \boldsymbol{1}_{\tau_{SM, \gamma} < t_0}]
\end{align*}
The monotone convergence theorem implies that $\EE [\tau_{SM, \gamma}]$ if finite
and
\begin{align*}
  (\alpha' - \alpha) \EE [\tau_{SM, \gamma}] & =  1 - 2\cdot \EE
  [(\boldsymbol{1}_{T_1 < T'_1} +\boldsymbol{1}_{S'_1 < S_1 < S_2 < S'_2})]\\
  & =  1 - 2 (p + p^2 q)
\end{align*}
This gives
$$
\EE [\tau_{SM,\gamma}]  = \frac{(1+pq) (p - q) + pq}{p - q} \, \tau_0   \ .
$$
Also we have
\begin{align*}
  \EE [N' (\tau_{SM, \gamma}) \cdot \boldsymbol{1}_{\tau_{SM, \gamma} < t_0}] & = \EE [N'
  (\tau_{SM, \gamma}) | \tau_{SM, \gamma} < t_0] \cdot \PP [\tau_{SM, \gamma} < t_0]\\
  & = \EE [N' (\tau_{SM, \gamma} \wedge t_0)] -\EE [N' (\tau_{SM, \gamma} \wedge t_0) |
  \tau_{SM, \gamma} > t_0] \cdot \PP [\tau_{SM, \gamma} > t_0]\\
  & = \alpha' \EE [\tau_{SM, \gamma} \wedge t_0] -\EE [N' (t_0)] \cdot
  \PP [\tau_{SM, \gamma} > t_0]\\
  & = \alpha' \EE [\tau_{SM, \gamma} \wedge t_0] - \alpha' t_0 \cdot \PP
  [\tau_{SM, \gamma} > t_0]\\
  & = \alpha' \EE [\tau_{SM, \gamma} \cdot \boldsymbol{1}_{\tau_{SM, \gamma} < t_0}]
\end{align*}
So, by the monotone convergence theorem again, we get $\EE [N' (\tau_{SM, \gamma})]
= \alpha' \EE [\tau_{SM, \gamma}]$. In the same way, $\EE [N (\tau_{SM, \gamma})] = \alpha
\EE [\tau_{SM, \gamma}]$.

 Finally, we observe that at the end of an attack cycle, the selfish miner has no orphan block 
 unless $T'_1 < T_1 < S_2 < S'_2$ and the second block of the honest miners is found by a fraction $1- \gamma$ of all 
 honest miners. This event occurs with a probability $p^2 q (1-\gamma)$ and, in this case the selfish, miner has exactly one orphan block. 
Therefore, we have 
\begin{align*}
  \EE [R(\tau_{SM, \gamma})] &= \EE \left[ N'\left( \tau_{SM, \gamma} \right)\right] b- p^2 q (1-\gamma) b\\
  &= \alpha' \EE [\tau_{SM, \gamma}]b  - p^2 q (1-\gamma)b
\end{align*}
and we get 
$$
\EE [R (\tau_{SM, \gamma})] =  \frac{(1+pq)(p - q) + pq}{p - q}\, qb -(1-\gamma) p^2q  \, b  \ .
$$

\end{proof}
\begin{remark*}
The introduction of $t_0$ in the above proof is a technical point so that we can use Doob's stopping time theorem.  
\end{remark*}

\subsection{Poisson games.}\label{sec_poisson}
Theorem \ref{profitsmg} is a variation of the following result. 
\begin{theorem}[Poisson games]
Let $N_1$ and $N_2$ be two independent Poisson processes with parameters $\alpha_1$ and $\alpha_2$ with $\alpha_1 > \alpha_2$ 
and $N_1(0) = N_2(0) = 0$. Then, the stopping time 
$$
\tau = \inf \{ t>0 ; N_1(t) = N_2(t) + 1\}
$$ 
is finite a.s. and integrable. Moreover, we have $\EE [\tau] = \frac{1}{\alpha_1 - \alpha_2}$,
$\EE[N_1(\tau)] =  \frac{\alpha_1}{\alpha_1 - \alpha_2}$, $\EE[N_2(\tau)] =  \frac{\alpha_2}{\alpha_1 - \alpha_2}$. 
\end{theorem}

\subsection{Apparent hashrate}
As discussed before, the good notion for objective function is the revenue of the miner by unit of time. However, we can also compute the revenue of the miner by unit of block. 
We compute the proportion $q'$ of blocks mined by the selfish miner in the official blockchain. This represents the apparent hashrate of the selfish miner.

\begin{proposition}\label{qprime}
We have $q' = \frac{((1 + pq) (p - q) + pq) q - (1 - \gamma) p^2 q (p - q)}{p^2 q + p - q}$
\end{proposition}

\begin{proof}

In all cases we observe that after one cycle of attack, the number of official blocks mined is always $\frac{N(\tau_{SM, \gamma}) + N'(\tau_{SM, \gamma}) + 1}{2}$. 
So, by the proof of Theorem \ref{profitsmg}, we get

\begin{align*}
  \EE \left [\frac{N(\tau_{SM, \gamma}) + N'(\tau_{SM, \gamma}) + 1}{2}\right ]  = 1+\frac{p^2 q}{p-q} \ .
\end{align*}

Hence, to get $n$ validated blocks in the official blockchain, the selfish miner needs to repeat his attack $X_n$ times with 

\begin{align}
  \EE [X_n]  =  \frac{n}{1+\frac{p^2 q}{p-q}} \ . \label{exn}
\end{align}

Therefore the average number of blocks mined by the selfish miner in a sequence of $n$ blocks is (we take $n\leq 2016$ to avoid a difficulty adjustment) 

\begin{align*}
  q' & = \frac{\EE [X_n]}{n} \cdot \frac{\EE [R(\tau_{SM, \gamma})]}{b}\\
&= \frac{((1 + pq) (p - q) + pq) q - (1 - \gamma) p^2 q (p - q)}{p^2 q + p - q}
\end{align*}

\end{proof}

We can rearrange this expression to get formula (8) for $R_{\text{pool}}$ from \cite{ES14},
$$
q'=\frac{q (1  - q)^2  (4 q + \gamma (1 - 2 q)) - q^3}{1 - q (1 + q (2 - q))} \ .
$$

\section{Selfish mining and difficulty adjustment}\label{sec_adj}
We examine now the impact of a difficulty's adjustment on the selfish mining strategy. 
After $n_0=2016$ blocks have been validated, the protocol modifies the speed of mining $\alpha$ and 
$\alpha'$ by a factor $\delta = \frac{\tilde{S}_{n_0}}{n_0 \tau_0}$ where $ \tilde{S}_{n_0}$ is the time needed
by the network to validate the sequence of $n_0$ blocks and $\tau_0=10 \text{ min}$.

\begin{proposition}\label{expd}
In presence of a single selfish miner with relative hashrate $q$, after the validation by the network of $n_0$ blocks, 
the parameter $\delta$ updating the difficulty's adjustment satisfies 
$$
\EE [\delta] = \frac{p - q + pq (p - q) + pq}{p^2 q + p - q} \ .
$$
\end{proposition}

\begin{proof}
Before the difficulty adjustment we have $\frac{1}{\alpha + \alpha'} = \tau_0$. So,  by (\ref{exn}) and Theorem \ref{profitsmg} with  $n_0 = 2016$, we have, 

\begin{eqnarray*}
 \EE[\tilde{S}_{n_0}] & = &  \EE[X_{n_0}] \cdot \EE[\tau_{SM, \gamma}]\\
&=& \frac{p - q + pq (p - q) + pq}{p^2 q + p - q}\cdot n_0 \tau_0
\end{eqnarray*}

\end{proof}

Note that $\EE[\delta]$ doesn't depend on $\gamma$ and that we have always $1\leq \EE[\delta] < 2$.

\begin{figure}[!ht]
   \includegraphics[height=6cm, width=9cm]{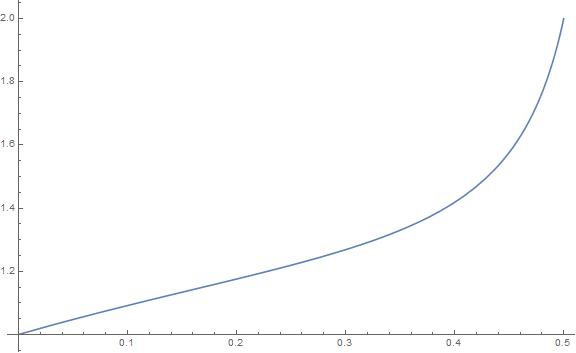}
   \caption{
   Difficulty adjustment $\EE[\delta]$ for $q \in[0, \frac{1}{2}]$.
   \medskip}
\end{figure}

We prove now that after a difficulty adjustment, the apparent hashrate and the revenue ratio coincide.

\begin{theorem}\label{profitpost}
After a difficulty adjustment, the revenue ratio of the selfish miner is $\Gamma(SM, \gamma) = \frac{q' b}{\tau_0}$
\end{theorem}

\begin{proof}
Before a difficulty adjustment, the speeds of validation $\alpha$ and $\alpha'$ were $\alpha = \frac{p}{\tau_0}$ and $\alpha' = \frac{q}{\tau_0}$. 
After a difficulty adjustment these quantities are both multiplied by a factor $\delta$. So, the revenue ratio 
$\Gamma (SM, \gamma)$ is also multiplied by $\delta$. 
Therefore, by Theorem \ref{profitsmg} and Proposition \ref{expd}, we get
$$\Gamma (SM, \gamma) = \left( qb - (1 - \gamma) \frac{p^2 q (p - q) b}{(1 + pq) (p - q) + pq} \right) 
\cdot \frac{p - q + pq (p - q) + pq}{p^2 q + p - q} \cdot \frac{b}{\tau_0}$$
After arranging this expression, we get $\Gamma (SM, \gamma) = \frac{q' b}{\tau_0}$.
\end{proof}

\begin{corollary}
After a period of difficulty adjustment, the selfish mining strategy becomes more profitable than staying honest  
if $\gamma > \frac{p - 2 q}{p - q} = \frac{1 - 3 q}{1 - 2 q}$.
\end{corollary}

\begin{proof}
The selfish mining strategy outperforms the strategy of "staying honest forever" if $P(SM, \gamma) >  \frac{q b}{\tau_0}$. 
The condition is equivalent to $q'>q$. Hence we get the result by Proposition \ref{qprime}.
\end{proof}
This condition is also equivalent to the condition $q>\frac{1 - \gamma}{3 - 2 \gamma}$ from \cite{ES14}. 
Note, and this point is crucial, that after a first period of difficulty's adjustment, the difficulty remains constant on average. 
So, the selfish mining strategy becomes more profitable than the honest strategy if the hashrates of the miners stay the same.

\section{Duration before the attack becomes profitable}

With our new model we can do some practical computations that are not possible with a pure Markov modeling. 
For example, we can compute the how long it takes to the selfish strategy to become profitable (since at the beginning 
it is not as we have proved). We compute the expected duration in this section.

\medskip

Since the goal of the selfish miner  is to decrease the difficulty, he will do it more efficiently  
starting his attack just after a difficulty adjustment so that he can have a maximal impact on the next 
one. So, we consider a miner that starts mining blocks following the selfish mining strategy just after a
difficulty adjustment.

\medskip

With the notations of Section \ref{sec_adj}, 
the time it takes to reach the next difficulty adjustment is $\tilde{S}_{n_0}=n_0 \tau_0 \delta$. 
Then, the attacker will have to wait an additional duration $t$ before his attack becomes profitable. 
We denote by $T_0=\tilde{S}_{n_0}+t$ this (break-even) duration.  

\begin{proposition}
We have $\EE[T_0] = \frac{q'  (\EE[\delta] - 1)}{q' - q} n_0 \tau_0$.
\end{proposition}

\begin{proof}
Let $\Gamma$ (resp. $\Gamma'$) be the revenue ratio before (resp. after) the difficulty adjustment.
We have seen that $\Gamma'=\Gamma . \EE[\delta]$ and $\Gamma'= q' \frac{b}{\tau_0}$ where $q'$ is the long-term apparent hashrate.
At $T_0$, the revenue of the selfish miner is exactly the same as the one he gets mining honestly from the beginning.
So, we have

$$
  \frac{q'}{\EE[\delta]}  (n_0 \tau_0 \EE[\delta]) + q' \EE[t] = q (n_0 \tau_0 \EE[\delta] + \EE[t])
$$

Hence, we get the result. 
\end{proof}

Note that $\lim_{q\to \frac{1}{2}} \frac{q'}{q}=\lim_{q\to \frac{1}{2}}\EE[\delta]=2$. 
So, $\lim_{q\to \frac{1}{2}} \EE[T_0] = 2 n_0 \tau_0$.

Except for $\gamma=1$, $\EE[T_0]$ first decreases with $q$ and then increases again up to $2$. 
For $\gamma=1$, $\EE[T_0]$ is always increasing with $q$.
Thus, for $\gamma$ fixed, there is a unique value $q_{min}<\frac{1}{2}$ such that the time it takes before the attack becomes profitable 
is minimum and for $q=q_{min}$, $\EE[T_0]<2 n_0 \tau_0$. For example, when $\gamma=\frac{1}{2}$, 
we have $\EE[T_0]$ minimum for $q=q_{min} \approx 43\%$ and in this case, $\EE[T_0] \approx 1.7 n_0 \tau_0\approx 23.8 \text{ days}$.
For $q=0.1$ and $\gamma=0.9$, we find that $\EE[T_0]\approx 5 n_0 \tau_0 = 10$ weeks 
and for $q=0.01$ and $\gamma=0.99$, we get $\EE[T_0] > 50 n_0 \tau_0\approx 100$ weeks, so approximately $1$ year and $11$ monthes...

\begin{figure}[!ht]
   \includegraphics[height=6cm, width=9cm]{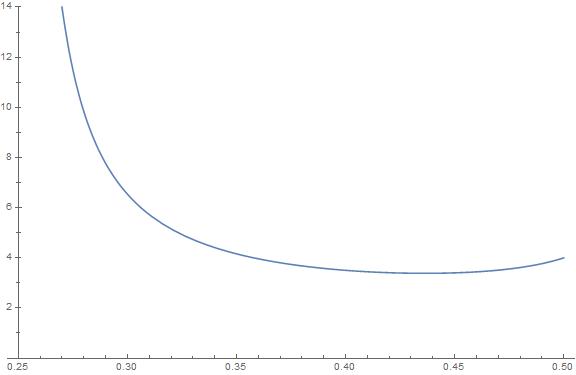}
   \put(5,5){\tiny{$q$}}
   \put(-265,175){\tiny{Weeks}}
   \caption{
   Graph of $\EE[T_0]$ in weeks for $\gamma=\frac{1}{2}$ and $q \in]\frac{1}{4},\frac{1}{2}[$.
   \medskip}
\end{figure}


\section{Pool formation}

In this section, we correct the proof given in Section 5 of \cite{ES14} for the growth of the
selfish miners pool. The situation is the following. There is a
pool of selfish miners with relative hashrate $q$, connectivity $\gamma$ and
apparent hashrate $q'>q$ after a difficulty adjustment. 
These authors claim the
following two facts:
\begin{enumerate}[(i)]
  \item \textit{Rational miners will prefer to join the selfish pool to
  increase their revenues.}\label{arg1}
  \medskip
  \item \textit{Members of the selfish mining pool are ready to accept new members, as this
  would increase their own revenue\label{arg2}}
\end{enumerate}
They give the criteria (Section 5 of \cite{ES14}) that (\ref{arg1}) results from  $q' > q$ 
and that (\ref{arg2}) results from  $\frac{\partial q'}{\partial q} > 1$. 
Note that with their notations, $R_{\text{pool}}$ is $q'$ and what they call ``pool size'' 
is the relative hashrate $q$. We show that both statements \ref{arg1} and \ref{arg2}  are true  
but not for the criteria given for \ref{arg2} is not.

In what follows, consider a modest honest miner H with a relative hashrate
$\varepsilon$ and a subgroup S of selfish miners who represents a fraction
$\lambda$ of the pool of selfish miners, i.e. the relative hashrate of S is
$\lambda q$.

We employ the notation $x'$ for the long-term apparent hashrate of a miner with relative
hashrate $x$ which follows the selfish mining strategy.

\subsection{When a pool of selfish miners is attractive for honest miners.}

On the long term, since the honest miner H represents only
$\frac{\varepsilon}{p}$ of the honest miners, his (long-term) revenue ratio
(i.e. his long-term apparent hashrate) is $\frac{\varepsilon}{p}  (1 - q')$.
Would he join the rogue pool, this quantity would then be equal to
$\frac{\varepsilon}{q + \varepsilon}  (q + \varepsilon)'$. Indeed, in this
case, the pool of selfish miners would have a relative hashrate equal to $q +
\varepsilon$ and the new member would represent $\frac{\varepsilon}{q +
\varepsilon}$ of this pool. So, H is attracted to the
pool of selfish miners if the following condition holds:  
$$
\frac{\varepsilon}{q + \varepsilon}  (q +
\varepsilon)' > \frac{\varepsilon}{p}  (1 - q')
$$ 
When $\varepsilon
\longrightarrow 0$, this leads to:
\begin{equation}
  \frac{q'}{q} > \frac{1 - q'}{1 - q} \label{ob1}
\end{equation}

It turns out that (\ref{ob1}) is equivalent to $q'>q$, and 
therefore, as claimed in \cite{ES14}, a pool of selfish miners
(with $q' > q$) is in principle attractive to honest miners.

\subsection{When a pool of selfish miners is willing to accept new members.}

Miner S is paid proportionally  to the quantity of hash power he brings
to the pool. So, if the relative hashrate of S is $\lambda q$ then his
long-term revenue ratio is $\lambda q'$. If H becomes a selfish miner, then,
the pool of selfish miners will have a total relative hashrate of $q +
\varepsilon$ and S will have now a fraction $\frac{\lambda q}{q +
\varepsilon}$ of this pool hashrate. Then, the new long-term revenue
ratio for S will be $\frac{\lambda q}{q + \varepsilon}  (q +
\varepsilon)'$. So, the condition for S to accept H is $\frac{\lambda q}{q +
\varepsilon}  (q + \varepsilon)' > \lambda q'$ which is equivalent to $\frac{
(q + \varepsilon)' - q'}{\varepsilon} > \frac{q'}{q}$. In the limit when
$\varepsilon \longrightarrow 0$, this gives the condition
\begin{equation}
  \frac{\partial q'}{\partial q} > \frac{q'}{q} \label{ob2}
\end{equation}

We have
$$
\frac{\partial \left( \frac{q'}{q} \right)}{\partial q} = \frac{1}{q}\, \left ( \frac{\partial q'}{\partial q} -\frac{q'}{q} \right )
$$
so condition (\ref{ob2}) is equivalent to $\frac{\partial \left( \frac{q'}{q} \right)}{\partial q} > 0$, 
which means that $\frac{q'}{q}$ must be increasing with $q$. It turns out that
(\ref{ob2}) is always satisfied as it can be checked from  the formula for
$q'$ given in Proposition \ref{qprime}. Therefore, as claimed by \cite{ES14}, a pool of selfish miners (with $q'
> q$) is always willing to accept new members but this follows from 
condition (\ref{ob2}) and not from the weaker condition $\frac{\partial
q'}{\partial q} > 1$.

\medskip

This analysis assumes that the individual miners have full knowledge about the selfish miner strategy. This means that the 
whole network is aware of the rogue behavior. But then it is very unlikely that countermeasures are not adopted. Thus we cannot
draw the type of conclusions than this authors claim in \cite{ES13}.

\begin{figure}[!ht]
   \includegraphics[height=6cm, width=9cm]{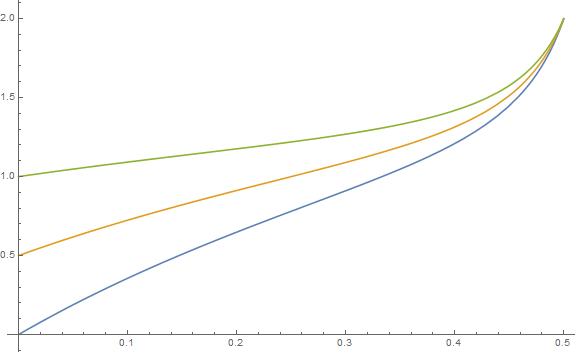}
   \put(-177,38){\tiny{$\gamma=0$}}
   \put(-220,49){\tiny{$\gamma=0.5$}}
   \put(-240,78){\tiny{$\gamma=1$}}
   \caption{
   Graphs of the functions $q'/q$ for $\gamma \in\{ 0, 0.5, 1\}$.
   \medskip}
\end{figure}

\section{A proposition to prevent selfish mining}

\subsection{The origin of the problem}
Basically, the attack exploits the difficulty adjustment law. The protocol underestimates the real 
hashing power in the network since
only the blocks that are in the (official) blockchain are taken into account.
The number of orphan blocks grows in the presence of a selfish miner and a significant amount of honest hashrate is lost. 
The average time used by the network to validate blocks increases. 
After 2016 blocks, a difficulty adjustment is done automatically
ignoring the production of orphan blocks. Despite the fact that the
total hashing power of the network remains the same, the new difficulty
is lower than it should be, and the block validation time decreases. 
So the revenue per unit of time of the selfish miner improves and makes the attack profitable.

\subsection{A new difficulty adjustment formula}

To mitigate this attack, the idea is to incorporate the count 
of orphan blocks in the difficulty adjustment formula. This can be implemented with 
miners indicating the presence of ``uncles' in the blocks they mine by including
their header and peers relaying this data. Only a signaling by honest miners will be enough. 
Nodes would not need to broadcast whole orphan blocks but only their
headers. It is possible to incentive miners to include proofs of existence of uncles 
in their blocks by including a rule that, in case of competition between two blocks 
with the same height, nodes should always broadcast \textit{ the block with the most proof-of-work} i.e., 
the block which includes the most proofs of existence of ''uncles". According to {\cite{PZ17}}, 
this rule would also be profitable to honest miners in case of blocks competition with selfish miners.
At the end of a period of $n_0 = 2016$ blocks validated by the network, the
new formula of difficulty adjustment would be
\begin{eqnarray}
  D_{\text{new}} & = & D_{\text{old}} \cdot \frac{(n_0 + n') \tau_0}{S_{n_0}} 
  \label{f}
\end{eqnarray}
where $n'$ is the total number of orphan blocks mined during this period of
time and $S_{n_0}$ is the time used by the network to validate the $n_0$
blocks (and evaluated with the formula $S_{n_0} = T_{n_0} - T_1$ where $T_i$
denotes the timestamp in the header of block $i$).

\subsection{Analysis of the formula}Let  $\omega$ be the average
number of orphan blocks observed during a period of $\tau_0 = 600$ sec. So, on
average, every $\tau_0$, there are $\omega$ orphan blocks and $(1 - \omega)$
non-orphan blocks. Only the last ones will add to the official blockchain. The time used by the network to grow the blockchain by
$n_0 = 2016$ blocks is then
$S_{n_0} = \frac{n_0 \tau_0}{1 - \omega}$. During this interval, we observe
$n' = \frac{n_0 \omega}{1 - \omega}$ orphan blocks. Thus we have $\frac{(n_0 + n')
\tau_0}{S_{n_0}} = 1$ and Formula (\ref{f}) cannot lead to a fall of difficulty.

\section{Conclusion}
Selfish mining is a trick that slows down the network and reduces the mining difficulty. 
The attack diminishes the profitability of honest miners
and the one of selfish miners before a difficulty adjustment. Selfish mining only becomes profitable after lowering the difficulty. 
Another way to achieve that would be to withdraw 
from the network and start mining another cryptocurrency with the same
hashing function. The existence of other cryptocurrencies with the 
same validation algorithm that allows to switch mining without cost is a vector of attack in itself.
When the attacker withdraws and comes back the mining difficulty will increase again after 2016 blocks 
unless the miner executes a selfish mining strategy. Then after a first difficulty's adjustment, the difficulty mining will stay constant on average. 

\medskip

Selfish mining is an attack on the Bitcoin protocol, but the arguments present in the literature do not properly justify the attack. T
hey lack of a proper analysis of profit and loss per unit of time. To compare the profitability of different mining strategies 
one needs to compute the average length of their cycles and their revenue ratio, that is a new notion introduced in this article. 

\medskip

The attack exploits a flaw in the difficulty adjustment formula. The parameter used to update the mining difficulty is supposed 
to measure the actual hashing power of the network. In the presence of a selfish miner, this is no longer true. 

\medskip

We have proposed a formula that corrects this anomaly by taking into account the production of orphan blocks.
We propose to reinforce the protocol which states that the official blockchain is the chain which contains the 
most proof-of-work by requiring peers to give priority to those containing the most proof-of-work (with ''uncles").

The proposed formula, if adopted, would not eliminate the possibility of selfish mining but it would make it non-profitable compared to honest mining 
even after a difficulty adjustment. So this will keep the individual incentives properly aligned in the protocol rules, as intended in the original 
inception of Bitcoin \cite{N08}.


\medskip

\textbf{Remark 1.}
 There are some selfish miners simulators available where the reader can test numerically the findings 
 in this article (see \cite{BH18} and \cite{MK2019}).
 
 \medskip
 
 \textbf{Remark 2.}
 The new theory developed in this article has been applied to  Stubborn and Trail Mining strategies presented in 
 \cite{NKMS2016}. The authors have solved completely the profitability problem for these strategies 
 in \cite{GPM2018-1} and \cite{GPM2018-2}, where we give close-form formulas. These formulas 
 are used to rigorously compare the profitability 
 of all these strategies and confirm and correct previous numerical studies (\cite{NKMS2016}). 
 Catalan numbers and the Catalan generating function appear naturally in these other problems.


\begin{thebibliography}{1}
  
  \bibitem[1]{BFGMN16}J.~Bonneau, E.~Felten, S.~Goldfeder, A.~Miller , 
  A.~Narayanan. \textit{Bitcoin and Cryptocurrency Technologies:
  A Comprehensive Introduction},  Princeton University Press, NJ,
  USA, 2016.
  
  
%
  
 \bibitem[2]{ES13}I.~Eyal, E.~G.~Sirer. \textit{Bitcoin is broken}, 
 hackingdistributed.com/2013/11/04/bitcoin-is-broken/ (accessed 1/2018), 2013.
  
  \bibitem[3]{ES14}I.~Eyal, E.~G.~Sirer. \textit{Majority is not
  enough: bitcoin mining is vulnerable},
  Int. Conf. Financial Cryptography and Data
  Security, Springer,  p.436-454, 2014.
  
  \bibitem[4]{GPM17}C.~Grunspan  and  R.~P\'erez-Marco. \textit{Double spend
  races}, International Journal of Applied and Theoretical Finance, Vol \textbf{21}, 8, 2018. 
  
  \bibitem[5]{GPM2018-1}C.~Grunspan  and  R.~P\'erez-Marco. \textit{On profitability of Stubborn Mining}, 
  ArXiv:1808.01041, 2018.
  
  \bibitem[6]{GPM2018-2}C.~Grunspan  and  R.~P\'erez-Marco. \textit{On profitability of Trailing Mining}, 
  ArXiv:1811.09322, 2018.
  
  
  \bibitem[7]{BH18}B.~Huisman. \textit{Selfish Mining and Difficulty Adjustments, 
  A Javascript Selfish Mining Simulator}. 
  www.greywyvern.com/code/javascript/selfishmining, 2018.


  \bibitem[8]{MK2019}M.~Khosravi. \textit{Selfish and Stubborn Mining Strategies}, https://armankhosravi.github.io/dirtypool, 2019.

  
  \bibitem[9]{N08}S.~Nakamoto. \textit{Bitcoin: a peer-to-peer electronic
  cash system}. \textit{Bitcoin.org}, 2008.
  
  \bibitem[10]{NKMS2016}K.~Nayak, S.~Kumar, A.~Miller, E.~Shi. \textit{Stubborn Mining: Generalizing 
  Selfish Mining and Combining with an Eclipse Attack}, 2016 IEEE Europ. Symp. on Security and Privacy, 2016.

  \bibitem[11]{R12}S.~Ross. \textit{ Introduction to Probability Models 10th Edition}.
  Academic Press Inc, 2012
  
  \bibitem[12]{SSZ16}A.~Sapirshtein, Y.~Sompolinsky, A.~Zohar, \textit{
  Optimal selfish mining strategies in bitcoin}, International Conference on Financial Cryptography and
  Data Security, Springer, p.515-532, 2016.

  \bibitem[13]{W17}R.~Wattenhofer. \textit{Distributed Ledger
  Technology: The Science of the Blockchain},  2nd Ed.,
  Create Space Independent Publishing Platform, 2017.
  
  \bibitem[14]{PZ17}R.~Zhang, B.~Preneel. \textit{Publish or perish: a
  backward-compatible defense against selfish mining in bitcoin}. In
  \textit{Topics in Cryptology - The Cryptographers Track at the RSA
  Conference 2017}, Springer,  p.277-292, 2017. 
\end{thebibliography}
\end{document}